\documentclass[a4paper,onecolumn,10pt,accepted=2024-05-22]{quantumarticle}
\pdfoutput=1

\usepackage[utf8]{inputenc}
\usepackage[english]{babel}
\usepackage{amsmath}
\usepackage{amssymb}
\usepackage[numbers]{natbib}

\usepackage[export]{adjustbox}
\usepackage{xcolor}
\usepackage{hyperref}
\hypersetup{
	bookmarksnumbered=true, % If Acrobat bookmarks are requested, include section numbers
	unicode=false, % non-Latin characters in Acrobat bookmarks
	pdfstartview={FitH}, % fits the width of the page to the window
	pdftitle={}, % title
	pdfauthor={}, % author
	pdfsubject={}, % subject of the document
	pdfcreator={}, % creator of the document
	pdfproducer={}, % producer of the document
	pdfkeywords={}, % list of keywords
	pdfnewwindow=true, % links in new window
	colorlinks=true, % false: boxed links; true: colored links
	linkcolor=blue, % color of internal links
	citecolor=blue, % color of links to bibliography
	filecolor=blue, % color of file links
	urlcolor=blue % color of external links
}
\usepackage{complexity}
\usepackage{amsthm}
\usepackage{thmtools}
\usepackage[capitalize,nameinlink]{cleveref} 
\usepackage{verbatim}
\usepackage{qcircuit}
\usepackage{tikz}
\usepackage{verbatim}
\usepackage{float}
\usepackage{multirow}
\usepackage{tabularx}
\usepackage{graphicx}

\usetikzlibrary{shapes.misc,shapes.geometric,arrows,positioning,calc,backgrounds,graphs,math,quotes}
\newcommand{\ket}[1]{| #1\rangle}        % ket vector
\newcommand{\bra}[1]{\langle #1|}        % bra vector
\newcommand{\kets}[1]{| #1 \rangle}        % small ket vector
\newcommand{\ketbra}[2]{| #1 \rangle\!\langle #2 |} % <x|y>
\newcommand{\ii}{\mathbb{I}}		% the I with two vertical lines
\newcommand{\norm}[1]{\left\| #1\right\|}        % norm
\newcommand{\norms}[1]{\| #1\|}        % small lines norm
\newtheorem{theorem}{Theorem}

\newcommand{\eq}[1]{Eq.~\hyperref[eq:#1]{(\ref*{eq:#1})}}
\renewcommand{\sec}[1]{\hyperref[sec:#1]{Section~\ref*{sec:#1}}}
\newcommand{\app}[1]{\hyperref[app:#1]{Appendix~\ref*{app:#1}}}
\newcommand{\tab}[1]{\hyperref[tab:#1]{Table~\ref*{tab:#1}}}
\newcommand{\fig}[1]{\hyperref[fig:#1]{Figure~\ref*{fig:#1}}}
\newcommand{\figa}[2]{\hyperref[fig:#1]{Figure~\ref*{fig:#1}#2}}
\newcommand{\figx}[2]{\hyperref[fig:#1]{Figure~\ref*{fig:#1}(#2)}}
\newcommand{\thm}[1]{\hyperref[thm:#1]{Theorem~\ref*{thm:#1}}}
\newcommand{\lem}[1]{\hyperref[lem:#1]{Lemma~\ref*{lem:#1}}}
\newcommand{\cor}[1]{\hyperref[cor:#1]{Corollary~\ref*{cor:#1}}}
\newcommand{\defn}[1]{\hyperref[def:#1]{Definition~\ref*{def:#1}}}
\newcommand{\alg}[1]{\hyperref[alg:#1]{Algorithm~\ref*{alg:#1}}}
\newcommand{\prob}[1]{\hyperref[prob:#1]{Problem~\ref*{prob:#1}}}
\newcommand{\threatM}[1]{\hyperref[threat:#1]{Threat Model \ref*{threat:#1}}}

\newcommand{\op}[1]{\operatorname{#1}}

\usepackage{graphicx}% Include figure files
\usepackage{dcolumn}% Align table columns on decimal point
\usepackage{bm}% bold math
\begin{document}
	\title{Trading T gates for dirty qubits in state preparation and unitary synthesis}
	
	\author{Guang Hao Low}
	\affiliation{Quantum Architectures and Computation, Microsoft Research, Washington, Redmond, USA}
        \affiliation{Azure Quantum, Microsoft, Washington, Redmond, USA}
	\author{Vadym Kliuchnikov}
	\affiliation{Quantum Architectures and Computation, Microsoft Research, Washington, Redmond, USA}
         \affiliation{Azure Quantum, Microsoft, Washington, Redmond, USA}
	\author{Luke Schaeffer}
	\affiliation{Quantum Architectures and Computation, Microsoft Research, Washington, Redmond, USA}
	\affiliation{Department of Electrical Engineering and 
 Computer Science, Massachusetts Institute of Technology, Cambridge, Massachusetts, USA}
 \affiliation{Joint Center for Quantum Information and Computer Science, University of Maryland, Maryland, College Park, USA}

	\date{2024-06-11}% It is always \today, today,
	
	\begin{abstract}
		Efficient synthesis of arbitrary quantum states and unitaries from a universal fault-tolerant gate-set e.g. Clifford+$\textsc{T}$ is a key subroutine in quantum computation. As large quantum algorithms feature many qubits that encode coherent quantum information but remain idle for parts of the computation, these should be used if it minimizes overall gate counts, especially that of the expensive $\textsc{T}$-gates. We present a quantum algorithm for preparing any dimension-$N$ pure quantum state specified by a list of $N$ classical numbers, that realizes a trade-off between space and $\textsc{T}$-gates. Our scheme uses $\mathcal{O}(\log{(N/\epsilon)})$ clean qubits and a tunable number of $\sim(\lambda\log{(\frac{\log{N}}{\epsilon})})$ dirty qubits, to reduce the $\textsc{T}$-gate cost to $\mathcal{O}(\frac{N}{\lambda}+\lambda\log{\frac{N}{\epsilon}}\log{\frac{\log{N}}{\epsilon}})$. This trade-off is optimal up to logarithmic factors, proven through an unconditional gate counting lower bound, and is, in the best case, a quadratic improvement in $\textsc{T}$-count over prior ancillary-free approaches. We prove similar statements for unitary synthesis by reduction to state preparation. Underlying our constructions is a $\textsc{T}$--efficient circuit implementation of a quantum oracle for arbitrary classical data. 
	\end{abstract}
	
	%	\pacs{Valid PACS appear here}% PACS, the Physics and Astronomy
	\maketitle
	\tableofcontents
	\section{Introduction} 
 Many quantum algorithms require coherent access to classical data, that is, data that can be queried in superposition through a unitary quantum operation. This property is crucial in obtaining quantum speedups for applications such as machine learning~\cite{LloydMohseniRebentrost2014}, simulation of physical systems~\cite{Berry2015Truncated,Low2016HamSim} and solving systems of linear equations~\cite{Harrow2009,Wang2018NonSparse}. 
	
The nature of quantum-encoded classical data is itself varied. For example, quantum data regression~\cite{Wiebe2012Fitting} queries a classical list of $N$ data-points $a_{x}$ through a unitary data-lookup oracle~\cite{Aharonov2003Adiabatic} 
 \begin{align}
     O\ket{x}\ket{0}=\ket{x}\ket{a_x}.
 \end{align}
 Other applications, particularly in quantum chemistry~\cite{Babbush2018encoding} instead access Hamiltonian coefficient data through a unitary $A\ket{0}=\ket{\psi}$ that prepares these numbers as amplitudes in a normalized quantum state, or as probabilities in a purified density matrix. Even more generally, the central challenge is synthesizing some arbitrary unitary $A\in\mathbb{C}^{N\times N}$ of which $k\le N$ columns are either partially or completely specified by a list of complex coefficients that is, say, provided on paper.
	
	Synthesis of these data-access unitaries is typically a dominant factor in the overall algorithm cost. In any scalable approach to quantum computation, all unitaries decompose into a universal fault-tolerant quantum gate set, such as Clifford gates $\{\textsc{H},\textsc S,\textsc{Cnot}\}$ and $\textsc{T}$ gates~\cite{Nielsen2004}. Solovay and Kitaev~\cite{Nielsen2004} were the first to recognize that any single-qubit unitary could be $\epsilon$-approximated using $\mathcal{O}(\log^c{(1/\epsilon)})$ fault-tolerant gates for $c=3.97$, which was later improved to $c=1$ ~\cite{Kliuchnikov2013synthesis,Ross2016Optimal}. By bootstrapping these results, it is well-known that a roughly equal number of $\mathcal{O}(kN\log{(N/\epsilon)})$~\cite{shende2006synthesis} Clifford and non-Clifford gates suffice for arbitrary dimensions. Notably, the total gate count scaling is optimal  in all parameters, following gate-counting arguments~\cite{Harrow2002}.
	
	The possibility that $\textsc{T}$ gates could be substantially fewer in number than the Clifford gates, however, is not excluded by known lower bounds. It is believed that fault-tolerant Clifford gates $\{\textsc{H},\textsc{S},\textsc{Cnot}\}$ will be cheap in most practical implementations of fault-tolerant quantum computation. In contrast, the equivalent cost of each fault-tolerant non-Clifford $\textsc{T}$ gates, implemented at machine precision, is placed at a space-time volume $\approx (225\;\textrm{logical qubits}) \times (10\;\textrm{Clifford depth})$ for realistic estimates~\cite{Litinski2019gameofsurfacecodes} based on $\ket{\textsc{T}}$ magic-state distillation at a physical error rate of $\approx 10^{-3}$. 
 {\renewcommand{\arraystretch}{1.1}
\begin{table*}
\centering
 \setlength\tabcolsep{1.0pt}
\resizebox{\textwidth}{!}{%
    \begin{tabular}{cc|rl|rl|rl|c}
        \multicolumn{2}{c|}{Output}&\multicolumn{2}{c|}{Qubits}&\multicolumn{2}{c|}{$\textsc{T}$ count}&\multicolumn{2}{c|}{$\textsc{T}$ Depth}& Valid $\lambda$\\
        \hline
        \multirow{3}{*}{$\ket{\psi}$} 
        &\cite{shende2006synthesis}&$\log{N}$&&$N\log{\frac{N}{\epsilon}}$&&$N\log{\frac{N}{\epsilon}}$ &&--
        \\
        &\cite{Sun2023State}& $\log{N}$&$+\lambda $& $N\log{\frac{N}{\epsilon}}$ && $\frac{N}{\log{N}+\lambda}\log{\frac{N}{\epsilon}}$&$+\log{N}\log{\frac{N}{\epsilon}}$ & $\mathcal{O}(\frac{N}{\log{N}\log\log{N}})\cup\Omega(N)$
        \\
        %&Gui et al.~\cite{Sun2023State}&$N$&$N\log{\frac{N}{\epsilon}}$&$\log{\frac{N}{\epsilon}}$
        %\\
        &Here & $\log{N}$&$+\lambda$ 
&$\lambda\log{\frac{N}{\epsilon}}$&$+\frac{N}{\lambda}\log{\frac{1}{\epsilon'}}$& $\frac{N}{\lambda}\log{\frac{1}{\epsilon'}}$&$+\log{N}\log{\frac{\lambda}{\epsilon'}}$
& $[\log{\frac{1}{\epsilon'}},N\log{\frac{1}{\epsilon'}}]$
        \\
        \hline
        \multirow{3}{*}{$U$}&\cite{shende2006synthesis}&$\log{N}$&&$N^2\log{\frac{N}{\epsilon}}$&&$N^2\log{\frac{N}{\epsilon}}$&&--
        \\
        &\cite{Sun2023State}& $\log{N}$&$+\lambda$ & $N^2\log{\frac{N}{\epsilon}}$&& $\frac{N^2}{\log{N}+\lambda}\log{\frac{N}{\epsilon}}$&$+N\log{N}\log{\frac{N}{\epsilon}}$& $\mathcal{O}(\frac{N}{\log{N}\log\log{N}})\cup\Omega(N)$
        \\
%&Clader et al.~\cite{Clader2022Blockencode}&$\log{N}$&$N^2\log{\frac{N}{\epsilon}}$&$N^2\log{\frac{N}{\epsilon}}$
        %\\

        &Here & $\log{N}$&$+\lambda$
        &$\lambda K\log{\frac{N}{\epsilon}}$&$+\frac{KN}{\lambda}\log{\frac{K}{\epsilon'}}$
        & $\frac{KN}{\lambda}\log{\frac{K}{\epsilon'}}$&$+K\log{N}\log{\frac{K\lambda}{\epsilon'}}$& $[\log{\frac{K}{\epsilon'}},N\log{\frac{K}{\epsilon'}}]$
        \\
        \hline
        \multirow{2}{*}{$\ket{\psi_\textrm{gb}}$}&\cite{Babbush2018encoding}&$\log{N}$&$+\;\;\log{\frac{1}{\epsilon}}$&$\log{\frac{N}{\epsilon}}$&$+N$&$N$&$+\log{\frac{1}{\epsilon}}$&--
        \\
        &Here& $\log{N}$&$+\lambda\log{\frac{1}{\epsilon}}$  &$\lambda\log{\frac{N}{\epsilon}}$&$+\frac{N}{\lambda}$& $\frac{N}{\lambda}$&$+\log{\frac{\lambda}{\epsilon}}$&$[1,N]$
    \end{tabular}}
    \caption{\label{tab:synthesis_cost_comparison} Big-$\mathcal{O}(\cdot)$ cost of preparing to error $\epsilon$ a dimension $N$ arbitrary quantum state $\ket{\psi}$~\cref{eq:state_prep}, quantum state  with garbage $\ket{\psi_\textrm{gb}}$~\cref{eq:state_prep_garbage}, and an $N\times N$ unitary $U$~\cref{eq:unitary_synthesis} where $K\le N$ columns are fully specified. Above, $\epsilon'=\frac{\epsilon}{\log{N}}$, and we have rescaled $\lambda b\rightarrow\lambda$ where appropriate to facilitate comparison.}
\end{table*}
}
	
	We present an approach to arbitrary quantum state preparation and unitary synthesis that focuses on minimizing the $\textsc{T}$ count. Unique to our approach is the exploitation of a variable number $\mathcal{O}(\lambda\log{(1/\epsilon)})$ of ancillary qubits, in a manner not considered by prior gate-counting arguments or algorithms. We find a $\mathcal{O}(\lambda)$ improvement in the $\textsc{T}$ count and circuit depth while keeping the Clifford count unchanged, excluding logarithmic factors. Most surprisingly, its benefit far exceeds the na\"{i}ve approach of applying these ancillary qubits to producing $\ket{\textsc{T}}$ magic-states for any $\lambda=\mathcal{O}(\sqrt{N})$, as seen in~\cref{tab:synthesis_cost_comparison}. In the best-case, the $\textsc{T}$ count of $\tilde{\mathcal{O}}(\sqrt{N})$ is a square-root factor smaller than prior art, such as for for preparing arbitrary pure states
	\begin{align}
	\label{eq:state_prep}
	\ket{\psi}=\sum_{x=0}^{N-1} \frac{a_x}{\|\vec{a}\|_2} \ket{x},\;\;\;
	\|\vec{a}\|_q=\Big(\sum_{j=0}^{N-1}|a_x|^q\Big)^{1/q},
	\end{align}
	Moreover, we prove this approach realizes an optimal ancillary qubit and $\textsc{T}$ count trade-off up to log factors.

	In particular, our approach is always advantageous as all but a logarithmic number $\mathcal{O}(\log{(N/\epsilon)})$ of qubits, independent of $\lambda$, may be dirty, meaning that they start in, and are returned to the same undetermined initial state. At first glance, the full quadratic speedup is not always desirable as any clean ancillary qubit, initialized in the $\ket{0}$ state, is a resource that may be better allocated to magic-state distillation. However, dirty qubits may not be used for magic-state distillation, and are a resource typically abundant in many algorithms, such as quantum simulation by a linear combination of unitaries~\cite{Berry2015Truncated}. Even in the most pessimistic scenario where no dirty qubits are available, a reduction in the overall execution time of the algorithm, including the effective cost of magic-state distillation, is possible.
	
	We also consider applications of our approach. For instance, a similar speedup to unitary synthesis 
	\begin{align}
	\label{eq:unitary_synthesis}
	U=\left(\sum^K_{k=0}\ket{u_k}\bra{k}\right)+\cdots,
	\end{align}
	where $K\le N$ columns are specified, follows by a well-known reduction based on Householder reflections. Improvements to state preparation with garbage
	\begin{align}
	\label{eq:state_prep_garbage}
	\ket{\psi_\textrm{gb}}=\sum_{x=0}^{N-1} \sqrt{\frac{a_x}{\|\vec{a}\|_1}} \ket{x}\ket{\mathrm{garbage}_x}.
	\end{align}
	relevant to the most advanced quantum simulation techniques~\cite{Low2016Qubitization,Low2017USA,Babbush2018encoding} are also presented in~\cref{sec:purified_density_matrix}.
	
	Underlying our $\textsc{T}$ gate scaling results is an improved implementation of a data-lookup oracle
	\begin{align}
	\label{eq:standard_oracle}
	O\ket{x}\ket{0}\ket{0}=\ket{x}\ket{a_x}\ket{\mathrm{garbage}_x}.
	\end{align}
	Note the attached garbage state may always be uncomputed by applying $O$ in reverse. We begin by describing our implementation of~\cref{eq:standard_oracle}, which we call a `$\textsc{SelectSwap}$' network, with costs outlined in~\cref{tab:cost_comparison}
 %\footnote{The $\textsc{T}$ count of $\textsc{SelectSwap}$ becomes suboptimal for $\lambda=\Omega(\sqrt{N/b})$. 
%
%For such large numbers of qubits, \cref{sec:indicator_function_construction} presents an alternate indicator function implementation that achieves nearly the same qubit--depth trade-off but with optimal $\textsc{T}$ count.
%
%However, we expect that most applications on reasonably-sized quantum computers will be bottlenecked by the production rate of magic states for implementing $\textsc{T}$ gates. Once optimal $\textsc{T}$ count is achieved, excess qubits are better spent on producing magics states to match the increased $\textsc{T}$ consumption rate that comes with a reduced $\textsc{T}$ depth. 
%}.
 Subsequently, we apply $\textsc{SelectSwap}$ to the state preparation problem using the fact that there exists classical data such that preparing any $\ket{\psi}$ requires only $\mathcal{O}(\op{polylog}(N))$ queries and additional primitive quantum gates~\cite{Aaronson2016complexity}. The reduction of unitary synthesis to state preparation is then described. Finally, we prove optimality of our approach through matching lower bounds, and discuss the results. 
	{
		\renewcommand{\arraystretch}{1.1}
		\begin{table}
  \centering
			\begin{tabular}{c|c|c|c}
				\shortstack{Operation\\$\;$}&\shortstack{Qubits\\$\;$}&\shortstack{$\textsc{T}$ count\\ $\le\cdot+\mathcal{O}(\log\cdot)$}&\shortstack{$\textsc{T}$ Depth\\$\mathcal{O}(\cdot)$}\\
				\hline
				$\textsc{Select}$&$b+2\lceil\log_2{N}\rceil$&$4N$&$N$
				\\
				$\textsc{Swap}$&$bN+\lceil\log_2{N}\rceil$&8$bN$&$\log{N}$ 
				\\
				$\textsc{SelSwap}$&$b\lambda +2\lceil\log_2{N}\rceil$&$4\lceil\frac{N}{\lambda}\rceil+8b\lambda$&$\frac{N}{\lambda}+\log{\lambda}$
				\\
				{\cref{fig:select}d}&$b(\lambda+1) +2\lceil\log_2{N}\rceil$&$8\lceil\frac{N}{\lambda}\rceil+32b\lambda$&$\frac{N}{\lambda}+\log{\lambda}$
			\end{tabular}
			\caption{\label{tab:cost_comparison}Upper bounds on cost of possible implementations of the data lookup oracle $O$ of~\cref{eq:standard_oracle}. See~\cref{sec:cnot_log_depth,sec:swap,sec:indicator_function_construction} for other variations, such as a linear-depth phase-incorrect version of \textsc{Swap} using $\le 4bN$ \textsc{T} gates and no additional ancillary qubits. Our results allow for a space-depth trade-off determined by a choice of $\lambda\in[1,N]$, with a minimized $\textsc{T}$ gate complexity of ${\mathcal{O}}(\sqrt{bN})$ by choosing $\lambda=\mathcal{O}(\sqrt{N/b})$. Note that $b\lambda$ qubits of the~\cref{fig:select}d implementation may be dirty.}
		\end{table}
	}

 \begin{figure*}[ht]
		\centering
  
 \setlength\tabcolsep{0.0pt}
		\begin{tabular}[t]{ll}
			\begin{tabular}[t]{ll}
                a) & b) \\
                \includegraphics[trim={0.17cm 0 4.2cm 0},clip,valign=t]{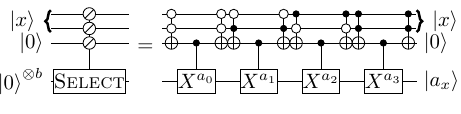} & \qquad\includegraphics[trim={0 0 3.5cm 0},clip,valign=t]{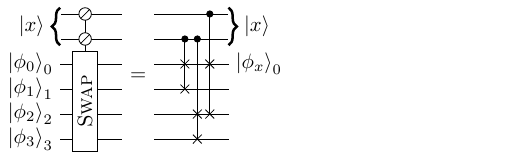} \vspace{0.2cm}\\
				c)&	d) \\
                \includegraphics[trim={0.17cm 0 4.2cm 0},clip,valign=t]{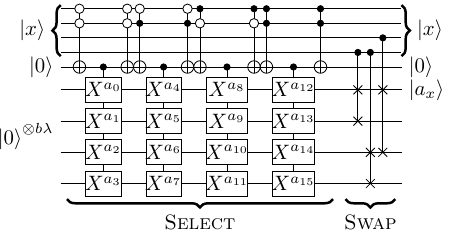}
                & \includegraphics[trim={0 0 3.5cm 0},clip,valign=t]{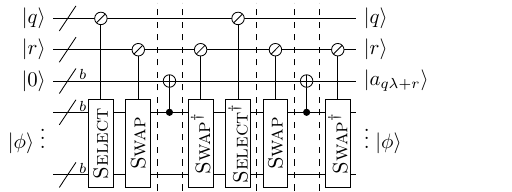}
			\end{tabular}
		\end{tabular}
		\caption{\label{fig:select}(a) Example $\textsc{Select}$ operator $\sum_{x=0}^{N-1}\ket{x}\bra{x}\otimes X^{a_x}$ with $N=4$. The symbol $\oslash$ indicates control by a number state. A naive decomposition of all multiply-controlled-$\textsc{Not}$s requires $\mathcal{O}(N\log{N})$ Clifford+$\textsc{T}$ gates and only one dirty qubit~\cite{Barenco1995gates}. Cancellation of adjacent gates can reduce this to only $\mathcal{O}(N)$~\cite{Childs2018Towards,Babbush2018encoding}, but by using additional $\lceil\log_2 N\rceil$ clean qubits. (b) Example $\textsc{Swap}$ network with $N=4$ using $\mathcal{O}(bN)$ Clifford+$\textsc{T}$. Any arbitrary state $\ket{\phi_x}_x$ in register index $x$ is swapped to the $x=0$ position. (c) The $\textsc{SelectSwap}$ network with $N=16, \lambda=4$ that combines the above two approaches. (d) Modification of $\textsc{SelectSwap}$ network that uses $2\lceil\log_2 N\rceil+b$ clean qubits and $b\lambda$ dirty qubits to implement the data-lookup oracle of~\cref{eq:standard_oracle} without garbage. We omit the $\textsc{Select}$ ancillary qubits for clarity.}
	\end{figure*}
 
	\section{Data-lookup oracle by a $\textsc{SelectSwap}$ network} 
	The unitary data-lookup oracle of~\cref{eq:standard_oracle} accepts an input number state $\ket{x}\in\mathbb{C}^N$ where $x\in\{0,1,\cdots,N-1\}\equiv[N]$, and returns an arbitrary $b$-bit number $a_x\in\{0,1\}^b$. Our approach combines a multiplexer implementation of $O$~\cite{Childs2018Towards}, called $\textsc{Select}$ and a unitary swap network $\textsc{Swap}$, with costs summarized in~\cref{tab:cost_comparison}.
	
	The $\textsc{Select}$ operator applies some arbitrary unitary $U_x$ controlled by the index state $\ket{x}$, that is
	\begin{align}
	\textsc{Select}=\sum_{x=0}^{N-1}\ket{x}\bra{x}\otimes U_x.
	\end{align}
	Thus $O$ is realized by choosing $U_x=X^{a_x}\equiv\otimes^{b-1}_{j=0}X^{a_{x,j}}$~\cite{Babbush2018encoding} to either be identity or the Pauli-$X$ gate depending on the bit string $a_x$. As described in~\cref{fig:select}a, the costs, excluding $\{U_x\}$, is $\mathcal{O}(N)$ Clifford+$\textsc{T}$ gates. As controlled-$X$ is Clifford too, an additional $\mathcal{O}(bN)$ Clifford gates are applied. These $\textsc{Cnot}s$ may be applied in logarithmic depth using an ancillary qubit free quantum fanout discussed in~\cref{sec:cnot_log_depth}.
	
	The unitary $\textsc{Swap}$ network moves a $b$-qubit quantum register indexed by $x$ to the $x=0$ register, controlled by the state $\ket{x}$. For any quantum states $\bigotimes_{x=0}^{N-1}\ket{\phi_x}_x$,
	\begin{align}
	\textsc{Swap}\left[\ket{x}\bigotimes_{x=0}^{N-1}\ket{\phi_x}_x\right]=\ket{x}\ket{\phi_x}_0\otimes \cdots,
	\end{align}
	where the remaining quantum states $(\cdots)$ in registers $x>0$ are unimportant. 
	
	As illustrated in~\cref{fig:select}b, this decomposes into a network of controlled-swap operators. As each controlled-swap operator decomposes into two $\textsc{Cnot}$s and one $\textsc{Toffoli}$, this network uses $\mathcal{O}(bN)$ Clifford+$\textsc{T}$ gates. An ancillary qubit free logarithmic-depth version of $\textsc{Swap}$ is discussed in~\cref{sec:swap}
	
	Our $\textsc{SelectSwap}$ network illustrated in~\cref{fig:select}c is a simple hybrid of the above two schemes. Similar to the $\textsc{Swap}$ approach, we duplicate the $b$-bit register $\lambda$ times, where $\lambda\in\{1,\cdots,N\}$ is an integer. For $\lambda$ that is not a power of $2$, we compute $\ket{x}\rightarrow\ket{q}\ket{r}$, which is the quotient $q=x/\lfloor\lambda\rfloor$ and remainder $r=x\mod{\lambda}$. This contributes an additive cost of $\mathcal{O}(\log{N}\log{\lambda})$ gates. $\textsc{Select}$ is controlled by $\ket{q}$ to write multiple values of $a_x$ simultaneously into these duplicated registers by choosing $U_x=\bigotimes^{\lambda-1}_{j=0} X^{a_{x N/\lambda+j}}$, where $x\in[\lfloor N/\lambda\rfloor]$. $\textsc{Swap}$ is then controlled $\ket{r}$ to move the desired data entry $\ket{a_x}$ to the output register. As the $\textsc{T}$ gate complexity of $\mathcal{O}(\lambda b+\frac{N}{\lambda})$ is determined only by the dimension of the $\textsc{Select}$ and $\textsc{Swap}$ control registers, this is minimized with value $\mathcal{O}(\sqrt{Nb})$ at $\lambda=\mathcal{O}(\sqrt{N/b})$.

	Importantly, all but $b+\lceil\log_2 N\rceil$ of the qubits may be made dirty using a simple modification shown in~\cref{fig:select}d. Then for any computational basis state $\ket{\phi}=\otimes_{r=0}^{\lambda-1}\ket{\phi_r}_r$, and any input state $\ket{x}=\ket{q}\ket{r}$, let us evaluate $\ket{0}\ket{\phi}$ at each dotted line:
	\begin{align}\nonumber
	\ket{0}\ket{\phi}
	&\rightarrow\ket{0}\ket{\phi_r\oplus a_x}_0\cdots
	\rightarrow\ket{\phi_r\oplus a_x}\ket{\phi_r\oplus a_x}_0\cdots\\\nonumber
	&\rightarrow\ket{\phi_r\oplus a_x}\ket{\phi}
	\rightarrow\ket{\phi_r\oplus a_x}\ket{\phi_r}_0\cdots
	\\
	&\rightarrow\ket{a_x}\ket{\phi_r}_0\cdots\rightarrow\ket{a_x}\ket{\phi}
	\end{align}
	By linearity, this is true for all quantum states $\ket{\phi}$.

	As the $\textsc{T}$ gate complexity begins to increase with sufficiently large $\lambda$, one may simply elect to not use excess available dirty qubits. However, continued reduction of the $\textsc{T}$ depth down to $\mathcal{O}(\log{N})$ might be a useful property. In~\cref{sec:indicator_function_construction} we discuss an alternate construction that achieves logarithmic $\textsc{T}$ depth and preserves the quadratic $\textsc{T}$ count improvement for larger $\lambda=\Omega(\sqrt{N})$.
	
	\section{Arbitrary quantum state preparation}
	Preparation of an arbitrary dimension $N=2^n$ quantum state $\ket{\psi}=\frac{1}{\|\vec{a}\|_2}\sum_{x\in\{0,1\}^n} a_x \ket{x}$ using the data-lookup oracle $O$ of~\cref{eq:standard_oracle} is well-known in prior art. The basic idea was introduced by~\cite{Grover2002creating}, and an ancillary-free implementation was presented in~\cite{shende2006synthesis}. We outline the inductive argument of~\cite{Aaronson2016complexity}, and evaluate its cost using our $\textsc{SelectSwap}$ implementation of $O$.
	
	For any bit-string $y\in\{0,1\}^w$ of length $w\le n$, let the probability that the first $w$ qubits of $\ket{\psi}$ are in state $\ket{y}$ be $p_y=\frac{1}{\|\vec{a}\|^2_2}\sum_{\mathrm{prefix}_w(x)=y}|a_x|^2$. Thus a single-qubit rotation $e^{-iY\theta}\ket{0}$ by angle $\theta=\cos^{-1}\sqrt{p_0}$ prepares the state $\ket{\psi_1}=\sqrt{p_0}\ket{0}+\sqrt{p_1}\ket{1}$, where $p_0$ is the probability that the first qubit of $\ket{\psi}$ is in state $\ket{0}$. We then recursively apply single-qubit rotations on the $(w+1)^{\mathrm{th}}$ qubit conditioned on the first $w$ qubits being in state $\ket{y}$. The rotation angles $\theta_y=\cos^{-1}{\sqrt{p_{y0}/p_y}}$ are chosen so that the state produced $\ket{\psi_{w+1}}$ reproduces the correct probabilities on the first $w+1$ qubits. 
	%For instance when $w=1$, we map
	%$
	%\ket{\psi_1}\mapsto \ket{\psi_2}=\sqrt{p_0}\ket{0}
	%\left(\sqrt{\frac{p_{00}}{p_0}}\ket{0}+\sqrt{\frac{p_{01}}{p_0}}\ket{1}\right)
	%+\sqrt{p_1}\ket{1}\left(\sqrt{\frac{p_{10}}{p_1}}\ket{0}+\sqrt{\frac{p_{11}}{p_1}}\ket{1}\right).
	%$
	
	These conditional rotations are implemented using a sequence of data-lookup oracles $O_1,\cdots, O_{n-1}$, where $O_w$ stores a $b$-bit approximation of all $\theta_y$ where $y\in\{0,1\}^w$. At the $w^{\mathrm{th}}$ iteration, 
	\begin{align}\nonumber
	\ket{\psi_w}&=\sum_{y\in\{0,1\}^w}\sqrt{p_y}\ket{y}
	%\\
	\underset{O_w}{\mapsto}
	\sum_{y\in\{0,1\}^w}\sqrt{p_y}\ket{y}\ket{\theta_y}
	\\\nonumber
	&\mapsto
	\sum_{y\in\{0,1\}^w}\sqrt{p_y}\ket{y}\left(\sqrt{\frac{p_{y0}}{p_y}}\ket{0}+\sqrt{\frac{p_{y1}}{p_y}}\ket{1}\right)\ket{\theta_y}
	\\%\nonumber
	&\underset{O^\dag_w}{\mapsto}
	\sum_{y\in\{0,1\}^{w+1}}\sqrt{p_y}\ket{y} 
	= \ket{\psi_{w+1}}.
	\end{align} 
	Note that we omit any garbage registers as they are always uncomputed. Also, the second line is implemented using $b$ single-qubit rotations each controlled by a bit of $\theta_y$. The complex phases of the target state $\ket{\psi}$ are applied to $\ket{\psi_n}$ by a final step with a data-lookup oracle storing $\phi_x=\op{arg}{[a_x/\sqrt{p_x}]}$. Thus $\mathcal{O}(b\log{N})$ single-qubit rotations are applied in total. 
	
	We implement these oracles with the $\textsc{SelectSwap}$ network of~\cref{fig:select}, using a fixed value of $\lambda$ for all $O_k$. A straightforward sum over the $\textsc{T}$ count of~\cref{fig:select} is $\mathcal{O}(b\lambda \log{(N)}+\frac{N}{\lambda})$, which is then added to the total $\textsc{T}$ count of $\mathcal{O}(b\log{(\frac{N}{\delta})})$ for synthesizing all single-qubit rotations each to error $\delta$ using the phase gradient technique~\cite{Gidney2018Addition}, outlined in~\cref{sec:pure_state}. The error of the resulting state $\ket{\psi'}$ produced is determined by the number of bits $b$ used to represent the rotation angles, in addition to rotation synthesis errors $\delta$. Adding these errors leads to
	\begin{align}
	\|\ket{\psi'}-\ket{\psi}\|=\underbrace{\mathcal{O}(\delta)}_{\text{rotation synthesis}}+\underbrace{\mathcal{O}(2^{-b}\log{N})}_{\text{bits of precision}}\le \epsilon,
	\end{align}
	which is bounded by $\epsilon$ with the choice $b=\Theta(\log{(\frac{\log{N}}{\epsilon})})$ and $\delta=\Theta(\epsilon)$. As a function of $\epsilon$, the total $\textsc{T}$ gate complexity is then $\mathcal{O}\left(\frac{N}{\lambda}+b(\lambda n+b)\right)$, where
	\begin{align}\label{eq:state_T_gates_1}
	b(\lambda n+b)
 &=\mathcal{O}\left(\lambda \log{\left(\frac{N}{\epsilon}\right)}\log{\left(\frac{\log{N}}{\epsilon}\right)}\right)\\\label{eq:state_T_gates_2}
 &=\mathcal{O}\left(\lambda \log^2{\left(\frac{N}{\epsilon}\right)}\right),
	\end{align}
	and is plotted in~\cref{fig:TcountStatePrep}.
	\begin{figure}
 \centering
		\includegraphics[width=0.5\columnwidth]{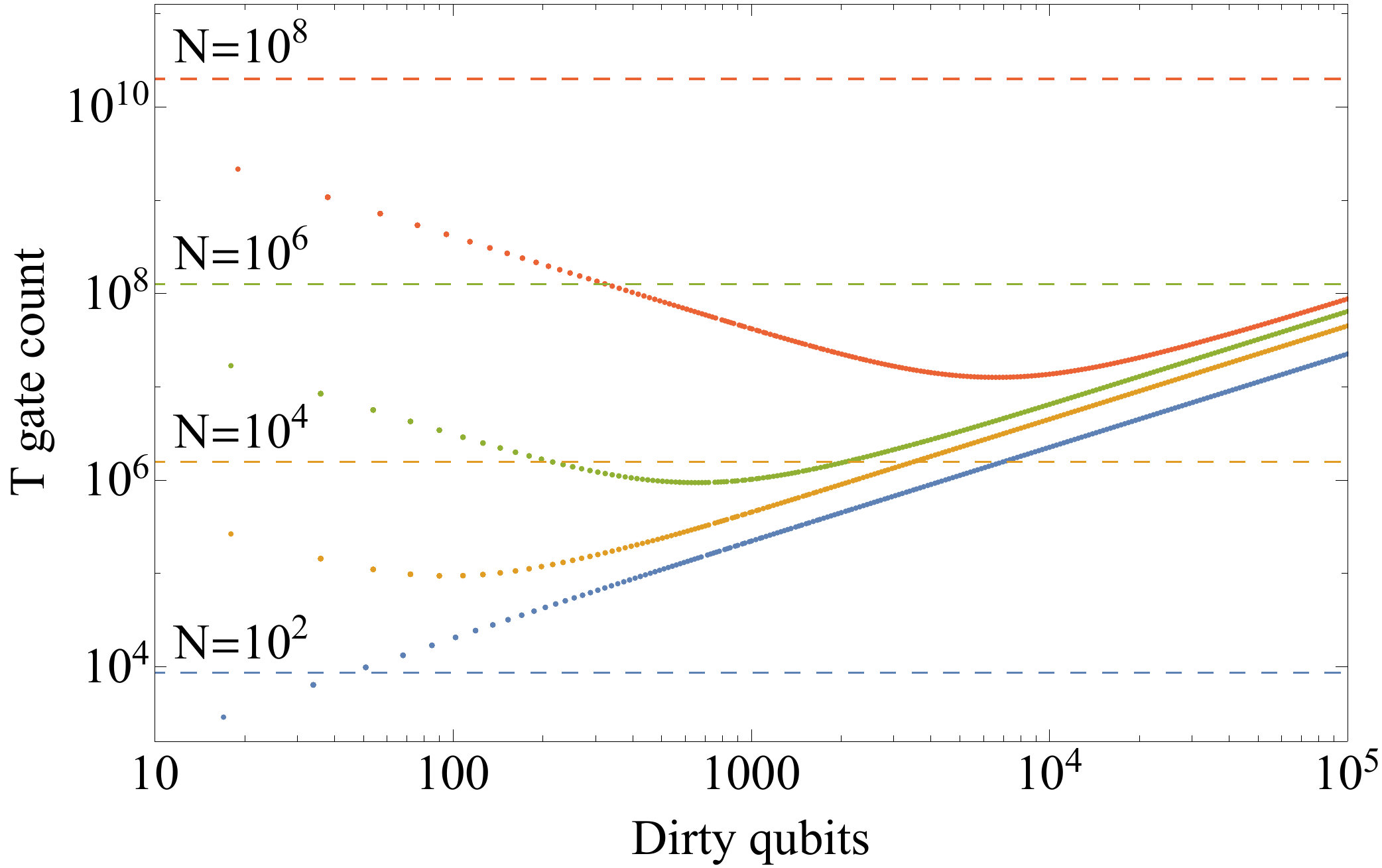}
		\caption{\label{fig:TcountStatePrep}$\textsc{T}$ gate count dependence on number of dirty qubits exploited for approximating an arbitrary quantum state of dimension $N=10^{\{2,4,6,8\}}$ to error $10^{-3}$ using our algorithm (dots) in comparison with the standard ancillary-free approach (dashed)~\cite{shende2006synthesis}. Note that at this error, $b\ge17$ qubits are required to represent each coefficient in binary, but this may be halved by randomization techniques~\cite{low2021halving}. Moreover, one may always use fewer than the maximum number of available dirty qubits.}
	\end{figure}

	\section{Unitary synthesis by state preparation}
	\label{sec:unitary_synthesis}
	The ability to prepare arbitrary quantum states enables synthesis of arbitrary unitaries $U\in\mathbb{C}^{N\times N}$. Given the matrix elements $\{\ket{u_k}\;|\;k\in[K]\}$ for the first $K$ columns of $U$, the 
	isometry synthesis problem is to find a quantum circuit that implements a unitary $V$ that approximates $U$ in the first $K$ columns to error
	$\norms{U-V}_{2,K} = \max _{ x < K } \norm{(U-V) \kets{x}}\le\epsilon$. 
	
	We use the Householder reflections decomposition~\cite{Householder1958} to find a $V$ that is a product of $K$ reflections 
	$\ii - 2\ketbra{v_k}{v_k}$ and a diagonal gate $\mathrm{diag}(e^{i\phi'_1}, \ldots, e^{i \phi'_K}, 1, \ldots, 1)$, for some set of quantum states $\ket{v_k}$. Note that this representation is not unique. The diagonal gate can be eliminated by using one ancillary qubit as discussed in~\cite{Kliuchnikov2013Reflections}. There, it suffices to implement the unitary 
	$W = \ketbra{0}{1}\otimes U + \ketbra{1}{0}\otimes U^\dagger$, which is equal to the product of reflections
	$\ii - 2\ketbra{w_k}{w_k}$ where $\ket{w_k} = (\ket{1}\otimes\ket{k} - \ket{0}\otimes\ket{u_k})/\sqrt{2}$ for $k\in[K]$. 
	
	Given a state preparation unitary $A_k\ket{0}=\ket{u_k}$, one can prepare state $\ket{w_k}$ starting from $\ket{0}\otimes \ket{e_1}$. Apply Hadamard to the first qubit, then a sequence of CNOT gates to prepare $\ket{1}\otimes\ket{k}+\ket{0}\otimes\ket{0}$, and finally apply controlled-$A_k$ negatively controlled on the first qubit. Note that when this method is applied to the synthesis of sparse isometries, the states being synthesized are again sparse.  Moreover, the cost of converting a state into a reflection doubles the number of non-Clifford gates.
	Thus the number of $\textsc{T}$ gates used to synthesize an isometry is twice that for all the Controlled-$A_k$ operations, and scales like $\mathcal{O}\left(K\left(\frac{N}{\lambda}+\lambda \log{(\frac{N}{\epsilon})}\log{(\frac{K\log{N}}{\epsilon})}\right)\right)$. 
	
	\section{Lower bound}
	We prove the optimality of our construction through a circuit counting argument. The most general circuit on $q$ qubits that uses $\Gamma$ $\textsc{T}$-gates has the canonical form~\cite{Gosset2014}
	$%\begin{align}
	C\cdot\prod^\Gamma_{j=1}e^{-i\pi P_j/8},
	$%\end{align} 
	where each $P_j$ is one of $4^q$ possible Pauli operators, and $C$ is one of $2^{\mathcal{O}(q^2)}$ possible Clifford operators. Thus the number of unique quantum circuits is at most 
	\begin{align}
	\label{eq:unique_circuits}
	\text{Unique quantum circuits} = \mathcal{O}( 4^{q \Gamma+\mathcal{O}(q^2)}).
	\end{align} 
	
	A lower bound on the qubit and $\textsc{T}$-gate complexity of the data-lookup oracle of~\cref{eq:standard_oracle} is obtained by counting the number of unique Boolean functions $f:[N]\rightarrow \{0,1\}^b$. As there $2^{b N}$ such functions, we compare with~\cref{eq:unique_circuits}. This leads to a lower bound on the space-\textsc{T} gate product
	\begin{align}
	q\Gamma = \Omega(b N-q^2).
	\end{align} 
	As the $\textsc{SelectSwap}$ complexity in~\cref{fig:select} is $q\Gamma=\mathcal{O}(\lambda^2 b^2+bN+\log{(N)}(1/\lambda+\lambda b))$, this is optimal up to logarithmic factors so long as the number of $\textsc{T}$-gates dominates the qubit count like $\lambda=o(\sqrt{N/b})$, which is the case in most quantum circuits of interest.  
	
	A similar lower bound on state preparation is obtained by counting the number of dimension-$N$ quantum states that a distinguishable with error $\epsilon$. Without loss of generality, we only count quantum states $\ket{\psi}\in\mathbb{R}^N$ with real coefficients. These states live on the surface a unit-ball $\mathbb{B}_N$ of dimension $N$, with area $\operatorname{Area}[\mathbb{B}_N]=\frac{2\pi^{N/2}}{(N/2-1)!}$. Let us now fix a state $\ket{\psi}$. Then the states $\ket{\chi}$ that satisfy $\|\ket{\psi}-\ket{\chi}\|\le \epsilon$ live inside a $\epsilon$-ball $\mathbb{B}_{N-1}$ with volume $\mathcal{O}(\operatorname{Vol}[\mathbb{B}_{N-1}] \epsilon^{N-1})=\mathcal{O}(\frac{\pi^{N/2}}{(N/2)!}\epsilon^{N-1})$~\cite{Mora2005Algorithmic}. Thus there are at least $\Omega(\frac{\operatorname{Area}[\mathbb{B}_N]}{\operatorname{Vol}[\mathbb{B}_{N-1}]\epsilon^{N-1}})=\Omega(\sqrt{N}\epsilon^{-N+1})$ quantum states. Once again by comparing with~\cref{eq:unique_circuits}, we obtain a $T$-gate lower bound of
	\begin{align}
	q \Gamma=\Omega(N\log{(1/\epsilon)}-q^2).
	\end{align}
	This also matches the cost of our approach in~\cref{eq:unique_circuits} up to logarithmic factors, so long as $\lambda = o(\sqrt{N/\log{(1/\epsilon)}})$. The total number of isometries within at least distance $\epsilon$ from each other can also be estimated using Lemma~4.3 on \href{https://arxiv.org/pdf/quant-ph/9508006.pdf#page=15}{Page~14} in \cite{Knill1995ApproxByQCirc}, and is roughly $\Omega((1/\epsilon)^{K N})$. An analogous argument can be made for state preparation with garbage by considering by considering the unit simplex instead of the unit ball. 
	
	Let us now establish a lower bound on state preparation that holds when measurements and arbitrary number of ancillae are used.
	For the purpose of the lower bound we also allow the use of post-selected measurements of multiple-qubit Pauli observables.
	Every preparation of an $n$-qubit state by Clifford+T circuit with ancillae and post-selected Pauli measurements can be rewritten~\cite{BCHK2018} as the following sequence of operations: 
	1) initialization of $\Gamma$ qubits into T state $\ket{0}+e^{i\pi/4}\ket{1}$; 
	2) post-selected measurement of $\Gamma - n$ commuting Pauli observables;
	3) application of a Clifford unitary on $\Gamma$ qubits.
	After the three steps first $\Gamma - n$ qubits are in a zero state and last $n$ qubits are in the state being prepared.
	Let us count the number of distinct states that can be prepared by the steps described above. 
	For the step two, there are at most $2\cdot4^{\Gamma}$ ways of choosing first Pauli observable and
	at most $4^{\Gamma}$ ways of choosing each of the remaining $\Gamma-n-1$ observables because
	each of them needs to commute with the first observable.
	Therefore, on step two we have at most $2\cdot4^{\Gamma(\Gamma-n)}$ choices of Pauli observables.
	For the step three, two distinct Clifford unitaries can lead to preparing the same state and
	counting total number of Clifford unitaries on $\Gamma$ qubits leads to an overestimate.
	The prepared $\ket{\psi}$ state is completely described by $4^n$ numbers $\alpha_P = \operatorname{Tr}(P\ket{\psi}\bra{\psi})$ where $P$ goes over all $n$-qubit Pauli matrices $\{I,X,Y,Z\}^{\otimes n}$.
	Let us count how many distinct $4^n$ dimensional vectors of $\alpha_P$ we can get on the step three by applying different Clifford unitaries $C$.
	Let $\rho'$ be a density matrix describing a state of all qubits after step two, then 
	$\alpha_P = \operatorname{Tr}(C\rho'C^\dagger (I_{\Gamma-n} \otimes P))$  
	which is equal to $\operatorname{Tr}(\rho' C^\dagger (I_{\Gamma-n} \otimes P) C)$. 
	We see that the vector of $\alpha_P$ is uniquely defined by action of $C$ on $2 n$ Pauli operators $Z_{\Gamma-n+k}$, $X_{\Gamma-n+k}$ for $k$ from $1$ to $n$.
	There are at most $2\cdot 4^{2\Gamma n}$ ways of choosing the action of $C$ on the listed Pauli operators. Therefore we can prepare at most $4\cdot4^{\Gamma(\Gamma+n)}$ distinct states. 
	This leads to the lower bound on the number of required \textsc{T} gates
	$\Omega\left(\sqrt{N \log N \log(1/\epsilon)}\right)$.
	
	\section{Conclusion}
	We have shown that arbitrary quantum states with $N$ coefficients, or unitaries with $KN$ values specified by classical data may be synthesized with a $\textsc{T}$ gate complexity that is an optimal $\tilde{\mathcal{O}}(\sqrt{N})$ reduction over prior art. As these subroutines are ubiquitous in many quantum algorithms we expect this result to be widely applicable. 
 
 We also expect our approach to be practical due to its almost exclusive usage of dirty qubits, which are typically abundant in larger quantum algorithms. Though our results are asymptotically optimal, constant factor and logarithmic improvements in costs could still be possible through careful optimization~\cite{Matteo2020QRAMresources,Kliuchnikov2022Spacetime}. For instance, our approach can be modified to use only $\mathcal{O}(1)$ additional clean qubits, but this increases the $\textsc{T}$ count by a logarithmic factor. As more limited trade-offs between $\textsc{T}$ gates and ancillary qubits are observed in other quantum circuits, such as for addition~\cite{Gidney2018Addition} or $\textsc{And}$~\cite{Barenco1995gates}, a major open question highly relevant to implementation in nearer-term quantum computers, is whether such a property could be generic for many other quantum circuits and algorithms.

	\subsection{Acknowledgements}
 We thank Nicolas Delfosse, Jeongwan Haah, Robin Kothari, Jessica Lemieux, Martin Roetteler, and Matthias Troyer for insightful discussions.

\section{Developments after 2018 release of preprint}
There have been numerous improvements in trade--offs~\cite{Sun2023State} between qubits, depth, and overall spacetime volume~\cite{Gui2024spacetime} of quantum circuits for state preparation and unitary synthesis~\cite{Sun2023State}, especially in connection to the block-encoding framework~\cite{Low2016Qubitization,Chakraborty2018BlockEncoding} for matrices represented by classical data~\cite{Clader2022Blockencode}. 
A notable new direction is error-resilient table-lookup~\cite{Haan2021QRAM}, which achieves logarithmic error scaling by using a large number of $\mathcal{O}(N)$ qubits. 
However, minimizing use of expensive non-Clifford gates is rarely a priority in these results, and almost all exhibit the same scaling for number of Clifford and non-Clifford gates.
We have updated~\cref{tab:synthesis_cost_comparison} with a comparison to recent work~\cite{Sun2023State} that claims an optimal qubit--depth trade--off.
To facilitate comparison, we rescale $\lambda$ and we use the tighter bound~\cref{eq:state_T_gates_1} on the \textsc{T} gate cost that was originally~\cref{eq:state_T_gates_2}.

Our \textsc{SelectSwap} architecture of table lookup remains state-of-art.
To date, all methods with an asymptotic advantage in non-Clifford gate cost for state preparation or unitary synthesis~\cite{Clader2022Blockencode} do so by reduction to \textsc{SelectSwap}.
The circuit implementation of \textsc{SelectSwap} has also seen some optimizations.
By allowing intermediate circuit measurements, the multi-target $\textsc{Cnot}_n$ gates in~\cref{sec:cnot_log_depth} can be implemented in constant, rather than logartihimic depth, though at the cost of using an additional $\mathcal{O}(n)$ clean qubits. 
Intermediate measurements also enable uncomputing the garbage register~\cref{eq:standard_oracle} of \textsc{SelectSwap} using $4\lceil\frac{N}{\lambda}\rceil+4\lambda$ $\textsc{T}$ gates and $\lambda+\lceil\log{(N/\lambda)}\rceil$ clean qubits, which is independent of $b$~\cite{Berry2019CholeskyQubitization}.  

The $\textsc{T}$ count reductions enabled by \textsc{SelectSwap} have been key to numerous state-of-art resource estimates for quantum computing applications, such as in chemistry~\cite{Berry2019CholeskyQubitization, vonBurg2020carbon, Lee2020hypercontraction, Bauer2020ChemistryReview}, and other classical-data-intense routines~\cite{Sanders2020Combinatorial,Clader2022Blockencode}.

\bibliographystyle{quantum}
	\bibliography{LowTStatePrepQuantum}

	\appendix
	\section{Purified density matrix preparation}
	\label{sec:purified_density_matrix}
	In some applications, particular in quantum simulation based on a linear combination of unitaries or qubitization~\cite{Berry2015Truncated,Low2016Qubitization}, it suffices to prepare the density matrix $\rho=\sum_{x=0}^{N-1}\frac{|a_x|}{\|\vec{a}\|_1}\ket{x}\bra{x}$ through a quantum state $\ket{\psi}=\sum_{x=0}^{N-1} \sqrt{\frac{a_x}{\|\vec{a}\|_1}} \ket{x}\ket{\mathrm{garbage}_x}$ of~\cref{eq:state_prep_garbage} where the number state $\ket{x}$ is entangled with some garbage that depends only on $x$. By allowing garbage, it was shown by~\cite{Babbush2018encoding} that strictly linear $\textsc{T}$ gate complexity in $N$ is achievable, using a $\textsc{Select}$ data-lookup oracle corresponding to the $\lambda=1$ case of~\cref{tab:cost_comparison}. We outline the original idea, then generalize the procedure using the $\textsc{SelectSwap}$ network, which enables sublinear $\textsc{T}$ gate complexity and better error scaling than the garbage-free approach. As density matrices have positive diagonals, we only consider the case of positive $a_x\ge0$.
	
	The original approach is based on a simple observation. By comparing a $b$-bit number state $\ket{a}$ together with a uniform superposition state $\ket{u_{2^b}}=\frac{1}{\sqrt{2^b}}\sum_{j=0}^{2^b-1}\ket{j}$ over $2^b$ elements, $\ket{a}$ may be mapped to
	\begin{align}
	\label{eq:comparator}
	\ket{a}\mapsto \ket{a}\Big(
	\sqrt{\frac{a}{2^b}}\ket{0}\ket{u_a}
	+\sqrt{\frac{2^b-a}{2^b}}\ket{1}\ket{u_{\ge a}}
	\Big),
	\end{align}
	where we denote a uniform superposition after the first $a$ elements by $\ket{u_{\ge a}}=\sum_{j=a}^{2^b-1}\frac{\ket{j}}{\sqrt{2^b-a}}$. This may be implemented using quantum addition~\cite{Cuccaro2004Adder}, which costs $\mathcal{O}(b)$ Clifford+$\textsc{T}$ gates with depth $\mathcal{O}(b)$.
	
	This observation is converted to state-preparation in four steps. First, the normalized coefficients $a_x\frac{N2^b}{\|a\|_1}\approx a_x'$ are rounded to nearest integer values such that $\|\vec{a}'\|_1=N2^b$. Second, the data-lookup oracle that writes two numbers $a_x''\in[2^b]$ and $f(x)\in[N]$ such that $a'_x=a''_x+\sum_{\mathrm{y}\in\{f^{-1}(x)\}}(2^b-a''_y)$. Thus
	\begin{align}
	\label{eq:QROM_data}
	O\ket{x}\ket{0}\ket{0}=\ket{x}\ket{a''_x}\ket{f(x)},
	\end{align}
	where we have omitted the irrelevant garbage state. Third, the oracle $O$ is applied to a uniform superposition over $\ket{x}$, and the comparator trick of~\cref{eq:comparator} is applied. This produces the state
	\begin{align}
	\sum_{x=0}^{N-1}\frac{\ket{x}}{\sqrt{N}}\ket{a''_x}\ket{f(x)}\left(\sqrt{\frac{a''_x}{2^b}}\ket{0}\ket{u_{a''_x}}
	+\cdots\ket{1}\ket{u_{\ge {a''_x}}}\right).
	\end{align}
	Finally, $\ket{f(x)}$ is swapped with $\ket{x}$, controlled on the $\ket{1}$ state. This leads to a state $\ket{\psi}=\sum_{x=0}^{N-1} \sqrt{\frac{a'_x}{\|\vec{a}'\|_1}} \ket{x}\ket{\mathrm{garbage}_x}$. After tracing out the garbage register, the resulting density matrix $\rho'$ approximates the desired state $\rho$ with trace distance
	\begin{align}
	\|\rho'-\rho\|_1=\mathcal{O}(2^{-b})\le \epsilon.
	\end{align}
	
	The $\textsc{T}$ gate complexity is then the cost of the data-lookup oracle of~\cref{eq:QROM_data} plus $\mathcal{O}(b)$ for the comparator of~\cref{eq:comparator}, plus $\mathcal{O}(\log{N})$ for the controlled swap with $\ket{f(x)}$. By implementing this data-lookup oracle with the $\textsc{SelectSwap}$ network, one immediately obtains the stated $\textsc{T}$ gate complexity of $\mathcal{O}(\lambda (b+\log{N})+\frac{N}{\lambda})=\mathcal{O}(\lambda\log{(N/\epsilon)}+\frac{N}{\lambda})$, where we choose $b=\mathcal{O}(\log{(1/\epsilon)})$. %This has a minimum value of $\mathcal{O}(\sqrt{N\log{(N/\epsilon)}})$ at $\lambda=\mathcal{O}(\sqrt{N/b})$.

	\section{Data-lookup oracle implementation details}
	\label{sec:data-lookup-details}
	In this section, we present additional details on the implementation of the data-lookup oracle. In particular, we discuss a multi-target $\textsc{Cnot}$ implementation in logarithmic depth without ancillary qubits in~\cref{sec:cnot_log_depth}, and a swap network $\textsc{Swap}$ with similar properties in~\cref{sec:swap}. Also evaluated is the $\textsc{T}$ count and Clifford depth of these implementations up to constant factors. We define the Clifford depth, to be the number of layers of two-qubit Clifford gates that cannot be executed in parallel, assuming all-to-all qubit connectivity. We also assume that each $\textsc{T}$ magic-state injection circuit has a Clifford depth of $1$.
	
	\subsection{Quantum fanout in logarithmic depth without ancillary qubits}
	\label{sec:cnot_log_depth}
	In this section, we construct a controlled-$\textsc{NOT}$ gate that targets $n$ qubits, that is,
	\begin{align}
	\label{eq:CNOT_n}
	\textsc{Cnot}_n=\ket{0}\bra{0}\otimes \mathbb{I}+\ket{1}\bra{1}\otimes X^{\otimes n}.
	\end{align}
	The most straightforward approach applies $n$ $\textsc{NOT}$ gates in sequence, each controlled by the same qubit. A slight modification results in logarithmic depth as shown in~\cref{tab:cost_comparison_cnot}.
	
	Given any number of $n$ qubits in states $\ket{x_{j}}$ for $j=0,1,\cdots,n-1$, one may use a ladder of $n-1$ controlled-$\textsc{NOT}$ gates to realize the transformation
	\begin{align}\nonumber
	&\ket{x_0}\ket{x_1}\ket{x_2}\cdots\ket{x_{n-1}}\\\nonumber
	\rightarrow&\ket{x_0}\ket{x_0\oplus x_1}\ket{x_2}\cdots\ket{x_{n-1}}\\\nonumber
	\rightarrow&\ket{x_0}\ket{x_0\oplus x_1}\ket{x_0\oplus x_1\oplus x_2}\cdots\ket{x_{n-1}}\\
	\rightarrow&\ket{x_0}\ket{x_0\oplus x_1}\ket{x_0\oplus x_1\oplus x_2}\cdots\ket{\oplus_{j=0}^{{n-1}}x_j}.
	\end{align}
	Let us call this unitary operation $\textsc{Ladder}_n$. 
	
	We now introduce a control qubit $\ket{z}$. One implementation of $\textsc{Cnot}_n$ is then obtained by applying $\textsc{Ladder}^\dag_n$, followed by a $\textsc{NOT}$ on $\ket{x_0}$ controlled by $\ket{z}$, and finally followed by $\textsc{Ladder}_n$. This has a Clifford depth of $2n-1$ as depicted below for the example $n=4$.
	\begin{center}
		\includegraphics[scale=1.0,trim={1.5cm 0 2cm 0}, clip]{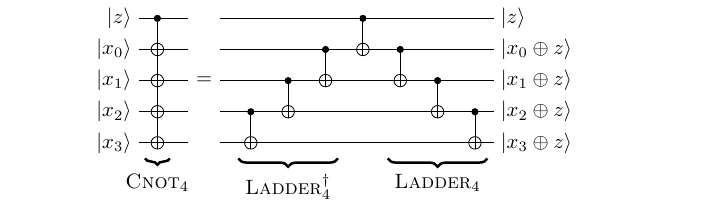}
	\end{center}
	
	By distributing the controls and targets above in a tree-structure as depicted below for the example $n=4$, the Clifford depth of $\textsc{Cnot}_n$ may be reduced to $2\lceil \log_2{(n)}\rceil+1$.
	\begin{center}
		\includegraphics[scale=1.0,trim={0 0 7cm 0}, clip]{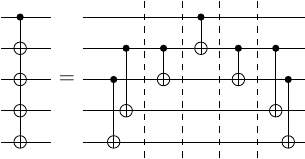}
	\end{center}
	As the control qubit $\ket{z}$ is only used once in the above circuit, a further reduction in depth is possible by repeatedly using it to apply additional multi-target $\text{Cnot}$ gates in each time-slice. Let us denote by $g(d)=2^{(d-1)/2}$ the maximum number of qubits targeted by the above circuit in a depth of $d$. Then the total number of qubits $n(d)$ targeted with this additional reduction satisfies the recurrence
	\begin{align}
	\label{eq:CNOT_best_depth}
	n(1)&=g(1),\\
	n(2)&=n(1)+g(1),\\
	\vdots&\\\nonumber
	n(d)&=n(d-1)+g(2\lceil d/2\rceil-1)
	=\begin{cases}
	(3\cdot 2^{(d-1)/2}-2),& d\;\text{odd},\\
	2(2^{d/2}-1),& d\;\text{even}.
	\end{cases}\ge 2(2^{d/2}-1)
	\end{align}
	Let us denote by $\text{D}[\textsc{Cnot}_n]$ the depth of this implementation of $\textsc{Cnot}_n$, which satisfies
	\begin{align}
	\label{eq:depth_cnot_n}
	\text{D}(\textsc{Cnot}_n)\le \left\lceil 2\log_2 \left(\frac{n+2}{2}\right)\right\rceil.
	\end{align}
	{
		\renewcommand{\arraystretch}{1.1}
		\begin{table}
  \centering
			\begin{tabular}{c|c|c|c}
				Approach &Clifford Depth $(d)$&Clifford Count&\shortstack{Volume \\ $=(n+1)d$}\\
				\hline
				{Linear}&$n$&$n$&$(n+1)d$
				\\
				{Logarithmic}&$2\lceil \log_2{(n)}\rceil+1$&$2n-1$&$(n+1)d$
				\\
				\cref{eq:CNOT_best_depth}&$\le 2\lceil \log_2{(n/2+1)}\rceil$&$\le 2n-1$&$(n+1)d$
			\end{tabular}
			\caption{\label{tab:cost_comparison_cnot}Different implementations of a controlled-$\textsc{NOT}$ gate~\cref{eq:CNOT_n} that targets $n$ qubits.}
		\end{table}
	}
	
	\subsection{Implementations of a \textsc{Swap} network}
	\label{sec:swap}
	In this section, we detail various implementations of the unitary swap network $\textsc{Swap}$ that moves an $b$-qubit quantum register indexed by $x\in\{0,\cdots,N-1\}$ to the position of the $x=0$ register, controlled by an index state $\ket{x}$. More precisely, for any set of quantum states $\bigotimes_{x=0}^{N-1}\ket{\phi_x}_x$ in the $b$-qubit register indexed by $x$,
	\begin{align}
	\textsc{Swap}\left[\ket{x}\bigotimes_{x=0}^{N-1}\ket{\phi_x}_x\right]=\ket{x}\ket{\phi_x}_0 \cdots,
	\end{align}
	where the final quantum states in registers indexed by $x>0$ are unimportant. Let us express the index ${x}\equiv x_0 x_1\cdots$ in binary, where $x_0$ is the smallest bit. Then it suffices to perform swaps between all pairs of registers indexed by $\{(i, i+2^j)\;|\;i\in 2^{j+1} \mathbb{N}_0  \}$, controlled by the $j^{\textrm{th}}$ qubit $\ket{x_j}$ of the index state $\ket{x}$, in the order of $j=0,1,\cdots,\lceil\log_2{N}\rceil-1$. 
	
	Each controlled pair-wise swap may be understood as a circuit $\textsc{CSwap}_n$ that swaps two $n$-qubit quantum registers in any state $\ket{\psi},\ket{\phi}\in\mathbb{C}^{2^n}$, controlled by a single qubit $\ket{z}\in\mathbb{C}^{2}$. That is,
	\begin{align}
	\label{eq:swap_controlled}
	\textsc{CSwap}_n\ket{z}\ket{\psi}\ket{\phi}=\begin{cases}
	\ket{z}\ket{\psi}\ket{\phi}, & z=0,\\
	\ket{z}\ket{\phi}\ket{\psi}, & z=1.
	\end{cases}
	\end{align}
	The overall cost of $\textsc{Swap}$ is then the sum of costs of $\textsc{CSwap}_n$, over all $j$, where $n=\lfloor\frac{N-1}{2^{j+1}}+\frac{1}{2}\rfloor$.
	
	We now consider different implementations considered realize trade-offs between Toffoli-gate count, circuit depth, and ancillary qubit usage, as summarized in~\cref{tab:cost_comparison_swap}.
	{
		\renewcommand{\arraystretch}{1.1}
  \centering
   \setlength\tabcolsep{1.5pt}
		\begin{table*}
			\begin{tabular}{c|c|c|c|l}
				Approach&Clifford Depth $(d_c)$&T count &T depth $(d_t)$&Volume $=(2n+1)(d_c+d_t)$\\
				\hline
				{Linear}&$4n+4$&$7n$&$4$&$8n^2 + \mathcal{O}(n)$
				\\
				{Logarithmic}&$4\text{D}[\textsc{Cnot}_{\lceil n/2\rceil}]+2$&$14n$&$8$&$\le 8n(\left\lceil 2\log_2 \left(\frac{n+2}{2}\right)\right\rceil+2.5)+\mathcal{O}(\log{n})$
				\\
				{Phase incorrect}&$\text{D}[\textsc{Cnot}_{n}]+4$&$4n$&$4$&$\le 2n(\left\lceil 2\log_2 \left(\frac{n+2}{2}\right)\right\rceil+8)+\mathcal{O}(\log{n})$
			\end{tabular}
			\caption{\label{tab:cost_comparison_swap}Different implementations of a controlled-swap between two $n$-qubit registers. The depth $\text{D}[\textsc{Cnot}_n]\le \left\lceil 2\log_2 \left(\frac{n+2}{2}\right)\right\rceil$ is from~\cref{eq:depth_cnot_n}.}
		\end{table*}
	}
	\subsubsection{\texorpdfstring{$\textsc{CSwap}_n$}{Swap} in linear depth without ancillary qubits}
	\label{sec:swap_linear_depth}
	It is simple to construct $\textsc{CSwap}_n$ with depth $\mathcal{O}(n)$ without any ancillary qubits. As the circuit that swaps two qubits is constructed from three $\textsc{Cnot}$ gates as follows,	
	\begin{center}
		\includegraphics[scale=1.0,trim={0 0 7cm 0}, clip]{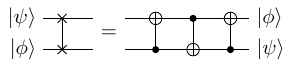}
	\end{center}
	a controlled-swap below is obtained by replacing the middle $\textsc{Cnot}$ with a Toffoli gate.
	\begin{center}
		\includegraphics[scale=1.0,trim={0 0 7cm 0}, clip]{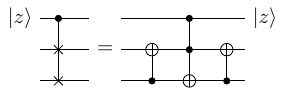}
	\end{center}
	A circuit that implements~\cref{eq:swap_controlled}, a swap between $n$ pairs of qubits, is then the above repeated $n$ times in sequence, each controlled by the same qubit as follows.
	\begin{center}
		\includegraphics[scale=1.0,trim={0 0 9cm 0}, clip]{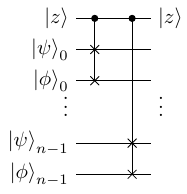}
	\end{center}
	
	\subsubsection{\texorpdfstring{$\textsc{CSwap}_n$}{Swap} in logarithmic depth without ancillary qubits}
	\label{sec:swap_log_depth}
	Constructing $\textsc{CSwap}_n$ with depth $\mathcal{O}(\log{n})$ without any ancillary qubits requires a little more thought. Let us consider a more general problem. Suppose we have an arbitrary unitary operator $V$ that is self-inverse, meaning $V^2=\mathbb{I}$ -- one may verify that the two-qubit swap satisfies this property. Our goal is to implement a multi-target controlled-$V$ gate on $n$ registers
	\begin{align}
	\label{eq:multi-target-V}
	\ket{0}\bra{0}\otimes \mathbb{I} + \ket{1}\bra{1}\otimes V^{\otimes n}. 
	\end{align}
	To begin, consider the following circuit identity, which is motivated by the `toggling` trick in~\cite{Haner2016factoring}. 
	\begin{center}
		\includegraphics[scale=1.0,trim={0 0 5cm 0}, clip]{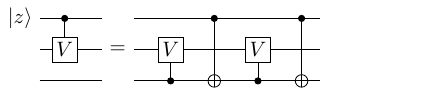}
	\end{center}
	Observe the the bottom qubit may be dirty -- its state does not affect the computation, and remains unchanged at the end of it. 
	
	Thus a multi-target controlled-$V$ on $n$ registers may be constructed by applying $n$ singly-controlled $V$ gates in parallel before and after a single multiply-controlled $\textsc{not}$ gate, using a total of $n$ extra dirty qubits as follows.
	\begin{center}
		\includegraphics[scale=1.0,trim={0 0 5cm 0}, clip]{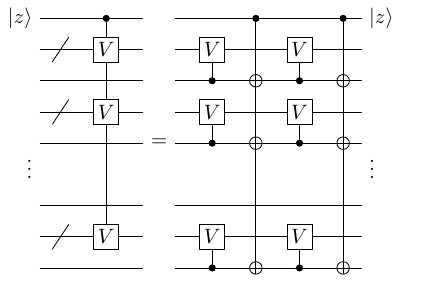}
	\end{center}
	As the additional qubits may be dirty, this is easily modified to use no ancillary qubits at all. Let us apply the multi-target $V$ on $\lceil n/2\rceil$ registers, using any $\lfloor n/2\rfloor$ qubits from the other registers as dirty qubits. When $n$ is odd, the topmost $V$ may be controlled directly by the $\ket{z}$ qubit. We then apply the same circuit on the remaining registers by using qubits in the initial targets as control qubits. In total, this uses at most $2n$ controlled-$V$ gates and two multiply-controlled $\textsc{not}$ gates on $\lceil n/2\rceil$ qubits, each with cost given by~\cref{tab:cost_comparison_cnot}.
	
	\subsubsection{\texorpdfstring{$\textsc{T}$}{T}-gate decomposition}
	Each controlled-swap may be decomposed into Clifford+$\textsc{T}$ gates using standard techniques. For instance, the standard synthesis of each Toffoli uses $7$ $T$-gates~\cite{shende2006synthesis}, as seen below. 
	\begin{center}
		\includegraphics[scale=1.0,trim={0 0 5.5cm 0}, clip]{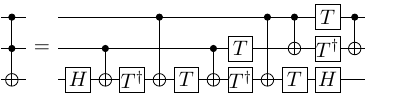}
	\end{center}
	Thus one might expect that $\textsc{CSwap}_n$ in logarithmic depth requires $14N$ $T$-gates. However, simple cancellations using the above decomposition reduces this to $10N$ $T$-gates. 
	
	A further reduction to just $4N$ $T$-gates is possible if we allow the output state to be correct up to a phase factor. The decomposition by~\cite{Barenco1995gates} below, using the gate $G=S^\dag\cdot H\cdot T\cdot H\cdot S$, approximates the Toffoli gate up to a minus sign on one of the matrix elements.
	\begin{center}
		\includegraphics[scale=1.0,trim={0 0 7cm 0}, clip]{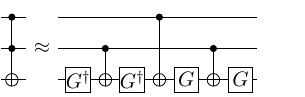}
	\end{center}
	Thus a controlled-swap is obtained by a simple modification as follows.
	\begin{center}
		\includegraphics[scale=1.0,trim={0 0 7cm 0}, clip]{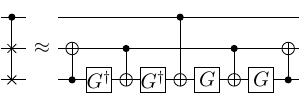}
	\end{center}
	Using this approximate controlled-swap, a version of $\textsc{CSwap}_n$ that is correct up to a $\pm 1$ phase may be obtained by replacing the middle $\textsc{Cnot}$ with $\textsc{Cnot}_n$, which may be implemented with logarithmic depth. As this incorrect phase may be absorbed into the garbage state of the data-lookup oracle, we may apply this phase-incorrect swap operation without loss of generality.
	
	\section{Indicator Function Construction}
	\label{sec:indicator_function_construction}
	
	There is an alternative construction based on implementing an \emph{indicator function}. That is, the function
	\begin{align}
	e^{(n)} \colon  \{ 0, 1 \}^{n} \to \{ 0, 1 \}^{N},
 \end{align}
	where $N = 2^{n}$, which maps a string $x \in \{ 0, 1 \}^{n}$ to an exponentially long encoding which is $1$ at position $x$ and zeros everywhere else.
	
	We observe that $e^{(n)}$ can be implemented with the following parameters.
	\begin{theorem}
		\label{thm:indicator}
		For $n \geq 1$, let $U_n$ be the unitary for the $n$th indicator function, $e^{(n)}$. That is, the circuit mapping  
		\begin{align}
		\ket{x} \ket{y} \mapsto \ket{x} \ket{e^{(n)}(x) \oplus y}
		\end{align}
		for all $x \in \{ 0, 1 \}^{n}$ and $y \in \{ 0, 1 \}^{N}$. There is a circuit computing $U_n$ in depth $O(\log^2 N)$ with $O(N)$ \textsc{T} gates without \emph{any} ancillary qubits. 
	\end{theorem}
	\begin{proof}
		When $n = 1, 2$ there are trivial circuits without ancillary qubits. For $n > 2$, let us suppose we have arbitrarily many clean ancillary qubits. We divide the input bits $x$ into two halves, $x_{hi}$ and $x_{lo}$, then recursively compute $e^{(n/2)}(x_{hi})$ and $e^{(n/2)}(x_{lo})$ in $O(2^{n/2}) = O(\sqrt{N})$ clean ancillary qubits each.
		
		Each output of $e^{(n)}(x)$ is the AND of a bit in $e^{(n/2)}(x_{lo})$ and a bit in $e^{(n/2)}(x_{hi})$. Each bit is used $O(\sqrt{N})$ times, so we can finish XORing $e^{(n)}(x)$ into $y$ with an array of Toffoli gates of $O(\sqrt{N})$. Alternatively, we can make $O(\sqrt{N})$ copies of each vector and do all the Toffoli gates simultaneously.
		
		Now suppose the ancillary qubits are dirty. Dirty qubits cannot store a qubit in the same way as clean qubits; since the initial state of the qubit is arbitrary, the information is encoded in the \emph{change} in state rather than the actual state. To implement a \textsc{Cnot} controlled on a dirty qubit, for instance, we apply the \textsc{Cnot} first, flip or not flip the qubit, then apply another \textsc{Cnot}. If there was no change, the two gates cancel, otherwise exactly one fires.
		
		Recall that the Toffoli gate has an implementation up to a faulty sign (which we can tolerate) as three \textsc{Cnot} interleaved with $G$ or $G^{\dag}$. As described above, we can implement each \textsc{Cnot} with a pair of \textsc{Cnot}s and, more importantly, a recursive call to the subroutine which populates that qubit, i.e., an indicator function. Actually, we use the previously described circuit for  $\textsc{Cnot}_{O(\sqrt{N})}$ instead of many individual \textsc{Cnot} gates for each Toffoli involving that bit.
		
		The depth cost for $\textsc{Cnot}_{O(\sqrt{N})}$ is $O(\log(\sqrt{N})) = O(n)$. Additionally, to compute $U_n$ we need 4 recursive calls to $U_{n/2}$. The depth satisfies the recurrence $D(n) = 4D(n/2) + O(n)$, so $D(n) = O(n^2) = O(\log^2 N)$.
		
		Finally, suppose we apply the array of Toffoli gates in \emph{two} layers instead of just one simultaneous layer. Since we are only applying Toffolis to half of the $\ket{y}$ register at any time, the other half can be used as dirty ancillary qubits. Asymptotically, this is much more than the $O(\sqrt{N})$ dirty qubits we need to compute $U_{n/2}$. Hence, we can implement $U_n$ without any additional qubits.
	\end{proof}
	
	It is wasteful to compute all of $e^{(n)}(x)$ in ancillary qubits, but it can be used to compute an arbitrary function. Better is to divide the input into two pieces,  similar to how we compute the indicator function.
	\begin{theorem}
		Suppose $f \colon \{ 0, 1 \}^{n} \to \{ 0, 1 \}^{b}$ is an arbitrary function, and let $U$ be the unitary mapping
		\begin{align}
		\ket{x} \ket{y} \mapsto \ket{x} \ket{y \oplus f(x)},
		\end{align}
		for all $x \in \{ 0, 1 \}^{n}$ and $y \in \{ 0, 1 \}^{b}$. 
		
		There is a circuit for $U$ of depth $O(\log^2 N + \log b)$, using $O(bN)$ dirty ancillary qubits and $O(\sqrt{N})$ $T$ gates. Alternatively, there is a circuit for $U$ using $O(\sqrt{bN})$ \textsc{T} gates and only $O(\lambda \sqrt{bN})$ dirty ancillary qubits, but with depth $O(\sqrt{bN} / \lambda + \log^2 N)$, for any $1 \leq \lambda \leq \sqrt{bN}$. 
	\end{theorem}
	\begin{proof}
		Divide the input into $n-k$ and $k$ bit pieces, where $k$ is to be determined later. Let $\hat{f} \colon \{ 0, 1 \}^{n-k} \to \{ 0, 1 \}^{b \times 2^k}$ be
		\begin{align}
		\hat{f}(x_{hi}) = \big( f(x_{hi} x_{lo}) \big)_{x_{lo} \in \{ 0, 1 \}^{k}},
		\end{align}
		the function which outputs the function for all possible values of the low order bits, given the high order bits. 
		
		Suppose, for the moment, that we have as many clean ancillary qubits as we want. We can na\"{i}vely compute $\hat{f}(x_{hi})$ by constructing $e^{(n-k)}(x_{hi})$, making $b \cdot 2^{k}$ copies, and computing the parity of some subset of bits of $e^{(n-k)}(x_{hi})$ for each output bit of $\hat{f}(x_{hi})$. Constructing $e^{(n-k)}(x_{hi})$ requires $O(2^{n-k})$ ancillary qubits, $O(2^{n-k})$ \textsc{T} gates and $O((n-k)^2)$ depth. The rest is done with $\textsc{Cnot}_{b \cdot 2^k}$ gates to make copies and $\textsc{Cnot}_{O(2^{n-k})}$ gates conjugated by Hadamards to compute parities. This uses many ancillary qubits--$O(2^{n-k})$ for the original vector times $O(b 2^{k})$ copies is $O(bN)$ qubits.
		
		Since the layers of \textsc{Cnot} gates compute (in the output bits) a linear function of the inputs, it is not difficult to adapt for dirty ancillary qubits. Just apply the linear function, flip the dirty ancillary qubits that are set to $1$, then apply the linear function again, appealing to the identity $T(x \oplus y) \oplus T(y) = T(x)$ for a linear function. Thus, whatever state $\ket{y}$ was in the dirty qubits, we still manage to compute $T(x)$. 
		
		We have shown how to compute $\hat{f}(x_{hi})$ and XOR it into dirty ancillary qubits. We also know how to compute $e^{k}(x_{lo})$ and XOR $b$ copies of it into dirty ancillary qubits. Think of $\hat{f}(x_{hi})$ as a $b \times 2^{k}$ matrix, then all that remains is to return the correct column of $\hat{f}(x_{hi})$. If we also think of $e^{k}(x_{lo})$ as a length $2^{k}$ column vector, then we are computing a matrix/vector product. The simplest way to do this is to make $b$ copies of $e^{k}(x_{lo})$ and execute $b$ vector/vector products in parallel, at a cost of $O(2^k)$ Toffoli gates for each one. 
		
		To compute the Toffoli gates on dirty ancillary qubits, we decompose them into \textsc{Cnot} gates and single qubit gates. The layers of \textsc{Cnot} gates are linear, so it is possible to compute each such layer with dirty ancillary qubits.
		
		Computing $\hat{f}(x_{hi})$ uses $O(bN)$ ancillary qubits, $O(2^{n-k})$ \textsc{T} gates, and has depth $O((n-k)^2 + k + \log b)$. The rest of the circuit has $O(b 2^{k})$ ancillas to store $\hat{f}(x_{hi})$ and/or copies of $e^{k}(x_{lo})$, and uses $O(b 2^{k})$ \textsc{T} gates in depth $O(k^2 + \log b)$. Since the \textsc{T} gate count is $O(2^{n-k} + b 2^{k})$, we set $k$ such that $2^{k} \approx \sqrt{N/b}$, and it becomes $O(\sqrt{b N})$.
		
		It is possible to trade off depth and number of ancillary qubits. We only need $O(2^{n-k})$ qubits to store $e^{(n-k)}(x_{hi})$ in the computation of $\hat{f}(x_{hi})$, if we are willing to compute parities for each of the $b 2^k$ output bits of $\hat{f}(x_{hi})$ one at a time, in depth $O(b 2^{k})$. More generally, we can use $O(\lambda 2^{n-k})$ ancillary qubits for any integer $\lambda \in [1, b 2^k]$ and use depth $O(b 2^{k} / \lambda)$. For the optimal \textsc{T} count we use the same setting of $k$, giving $O(\sqrt{b N } / \lambda + \log^2 N)$ depth with $O(\lambda \sqrt{b N})$ ancillary qubits and $O(\sqrt{b N})$ \textsc{T} gates.
		
	\end{proof}
	
	\section{Pure-state preparation implementation details}
	\label{sec:pure_state}
	The approach by Shende, Bullock, and Markov~\cite{shende2006synthesis} synthesizes a unitary $A$ that prepares a pure state $A\ket{0}=\sum_{x=0}^{N}\frac{a_x}{\|\vec{a}\|_2}\ket{x}=\ket{\psi}$ with arbitrary coefficients in $N=2^n$ dimensions. The underlying circuit, illustrated below for the example of $N=8$ for positive coefficients,
	\begin{center}
		\includegraphics[scale=1.0,trim={0 0 8cm 0}, clip]{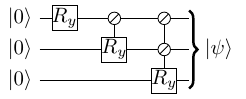}
	\end{center}
	is built from $j\in[n]$ multiply-controlled arbitrary single qubit rotations, $U_j$ where
	\begin{align}
	U_j=\sum_{x=0}^{2^{j}-1}\ketbra{x}{x}\otimes e^{i2\pi \theta_{j,x}Z},
	\end{align}
	for some set of rotation angles $\theta_{j,x}$. Note that it suffices to consider $Z$-phase rotations as rotations about the $X,Y$ Pauli operators are equivalent up to a single-qubit Clifford similarity transformation. Each multiplexor is applied twice -- once to create a pure state with the right probabilities $|a_x|^2$, and once to apply the correct phase $e^{i\arg{[a_x]}}$. Below we describe how $U_j$ may be implemented using the data-lookup oracle of~\cref{eq:standard_oracle}
	\begin{align}
	O\ket{x}\ket{0}\ket{0}=\ket{x}\ket{a_x}\ket{\mathrm{garbage}_x},
	\end{align}
	and evaluate the overall error and cost of state preparation.
	
	\subsection{Multiply-controlled phase gate from data lookup oracles}
	Consider a multiply-$n$-controlled arbitrary single qubit rotation
	\begin{align}
	\label{eq:multiplex-rotations-original}
	U=\sum_{x=0}^{2^{n}-1}\ketbra{x}{x}\otimes e^{i2\pi \theta_{x}Z},
	\end{align}
	\begin{center}
		\includegraphics[scale=1.0,trim={0 0 8.5cm 0}, clip]{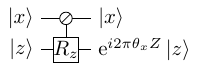}
	\end{center}
	where each rotation angle $\theta_x \in[0,1)$. Given a number state $\ket{x}$ and an arbitrary single-qubit state $\ket{z}$,
	this unitary performs a controlled-rotation
	\begin{align}
	U\ket{x}\ket{z}=\ket{x}e^{i2\pi \theta_{x}Z}\ket{z}.
	\end{align}
	Each rotation angle has a binary expansion
	\begin{align}
	\theta_x = \sum_{k=0}^{\infty}\frac{a_{x,i}}{2^{k+1}}.
	\end{align}
	By truncating to the above to $b$-bits of precision, we obtain an integer approximation $a_x$ of $\theta_x$ where
	\begin{align}
	a_x =2^b\sum_{k=0}^{b-1}\frac{a_{x,k}}{2^{k+1}},\quad \left|\frac{a_x}{2^b}-\theta_x\right|<\frac{1}{2^b}.
	\end{align}
	
	Let us encode these values of $a_x$ into the data-lookup oracle, and express its $\textsc{T}$ cost as the function $f(n,b)$. Its output is then
	\begin{align}
	O\ket{x}\ket{0}^{\otimes b}\ket{0}=\ket{x}\left[\bigotimes_{k=0}^{b-1}\ket{a_{x,k}}\right]\ket{\mathrm{garbage}_x},
	\end{align}
	where we explicitly represent the number state $\ket{a_x}$ in terms of its component qubits. 
	
	\subsubsection{Approach using arbitrary single-qubit synthesis}
	One possible approximation, call it $U'$, of $U$ in~\cref{eq:multiplex-rotations-original} applies the single-qubit rotation
	$e^{i\pi Z/2^{k}}$ to the target state $\ket{z}$, controlled by the state $\ket{a_{x,k}}$. The garbage register is then uncomputed by running $O$ in reverse. Explicitly, this circuit realizes the transformation
	\begin{align}\nonumber
	\label{eq:multiplex-rotation-sequence}
	&\ket{x}\ket{0}^{\otimes b}\ket{0}\ket{z} \\\nonumber
	\rightarrow& \ket{x}\left[\bigotimes_{k=0}^{b-1}\ket{a_{x,k}}\right]\ket{\mathrm{garbage}_x}\ket{z}\\\nonumber
	\rightarrow& \ket{x}\left[\bigotimes_{k=0}^{b-1}\ket{a_{x,k}}\right]\ket{\mathrm{garbage}_x}e^{i\pi \sum^{b-1}_{k=0} a_{x,k}Z/2^{k}}\ket{z}\\\nonumber
	= & \ket{x}\ket{a_x}\ket{\mathrm{garbage}_x}e^{i2\pi a_x/2^{b}Z}\ket{z}\\
	\rightarrow& \ket{x}\ket{0}^{\otimes b}\ket{0}e^{i2\pi a_x/2^{b}Z}\ket{z}.
	\end{align}
	Now, each controlled-arbitrary phase rotation decomposes into $2$ arbitrary single-qubit rotations, and $\textsc{CNot}$ gates -- Note that a decomposition into $1$ arbitrary single-qubit rotation is possible if we modify the above for the range $\theta_x\in(-1/2,1/2)$, but the explanation is slightly more complicated.	As any arbitrary single-qubit rotation is approximated to error $\epsilon$ using $\mathcal{O}(\log{(1/\epsilon)})$ $\textsc{T}$ gates~\cite{Kliuchnikov2013Ancilla}, we may use a triangle inequality to bound the error to
	\begin{align}\nonumber
	\|U-U'\|&\le \left\|\sum_{x=0}^{2^{n}-1}\ketbra{x}{x}\otimes (e^{i2\pi \theta_{x}Z}-e^{i2\pi a_x/2^bZ})\right\|+2b\epsilon\\\nonumber
	&\le\max_{|y|<1/2^b}\left|(e^{i2\pi y}-1)\right|+2b\epsilon \\
	&<\frac{2\pi}{2^b}+2b\epsilon.
	\end{align}
	Thus for any target error $\|U-U'\|=\delta$, we may solve for the $b,\epsilon$ parameters and bound the number of $\textsc{T}$ gates required to
	\begin{align}
	f(n,b)+\mathcal{O}(\log^2{(1/\delta)}), \quad b=\mathcal{O}(\log{(1/\delta)}).
	\end{align}
	
	\subsubsection{Approach using phase gradients}
	A more efficient approach~\cite{Gidney2018Addition} uses a Fourier state as a resource
	\begin{align}
	\mathcal{F}=\frac{1}{\sqrt{2^b}}\sum_{k=0}^{2^b-1}e^{-2\pi i k/2^b}\ket{k}.
	\end{align}
	combined with a reversible adder
	\begin{align}
	\textsc{Add}\ket{x}\ket{y}&=\ket{x}\ket{{y+x}\mod{2^b}}
	\end{align}
	Observe that
	\begin{align}
	\textsc{Add}\ket{x}\ket{\mathcal{F}}=e^{2\pi i x/2^b}\ket{x}\ket{\mathcal{F}}.
	\end{align}
	Thus the controlled adder, known to cost $\mathcal{O}(b)$ $\textsc{T}$ gates, 
	\begin{align}
	\textsc{Cadd}=\textsc{Add}\otimes\ket{0}\bra{0}+\textsc{Add}^\dag\otimes\ket{1}\bra{1}
	\end{align}
	realizes the controlled-phase rotation
	\begin{align}
	\textsc{Cadd}\ket{x}\ket{\mathcal{F}}\ket{z}=\ket{x}\ket{\mathcal{F}}e^{2\pi i x/2^b Z}\ket{z}.
	\end{align}
	
	Thus another possible approximation, call it $U''$, uses this adder, controlled by the target state $\ket{z}$, to the registers containing the desired phase rotation $\ket{a_x}$, and the Fourier state. This realizes the same transformation in~\cref{eq:multiplex-rotation-sequence}, using the circuit depicted below.
	%	\begin{widetext}
	\begin{center}
		\includegraphics[scale=0.85]{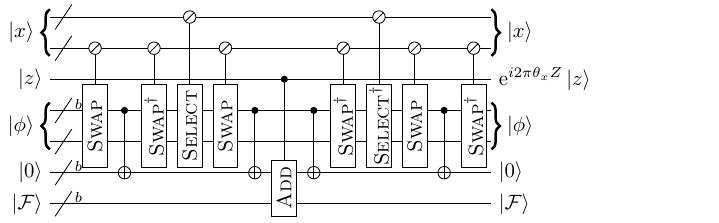}
	\end{center}
	%	\end{widetext}
	Assuming that the Fourier state is prepared perfectly, this approximates $U$ with error
	\begin{align}
	\|U-U''\|&\le \left\|\sum_{x=0}^{2^{n}-1}\ketbra{x}{x}\otimes (e^{i2\pi \theta_{x}Z}-e^{i2\pi a_x/2^bZ})\right\|\le\max_{|y|<1/2^b}\left|(e^{i2\pi y}-1)\right|<\frac{2\pi}{2^b}.
	\end{align}
	The state preparation unitary $A''$ applies $n$ such approximations $U_j''$ to the multiply-controlled rotations $U_j$, leading to an overall error bounded by
	\begin{align}
	\left\|\prod_{j=0}^{n-1}U_j-\prod_{j=0}^{n-1}U_j''\right\|\le \sum_{j=0}^{n-1}\|U_j-U_j''\|<\frac{2n\pi}{2^b}.
	\end{align}
	Note that the cost of approximating the Fourier state to error $\epsilon$ has a T cost of $\mathcal{O}(b\log{(1/\epsilon)})$. This imperfect Fourier state $\ket{\mathcal{F}'}$ contributes a one-time error of $\epsilon$, following the inequality
	\begin{align}
	\|W(\ket{\mathcal{F}'}\otimes\mathbb I)-(W\ket{\mathcal{F}}\otimes\mathbb{I}\|\le \|\ket{\mathcal{F}'}-\ket{\mathcal{F}}\|\le \epsilon,
	\end{align}
	for any arbitrary unitary operator $W$. Thus $A''$ prepares the state $\ket{\psi''}=A''\ket{0}\ket{\mathcal{F}'}$, which approximates $\ket{\psi}$ with error
	\begin{align}
	\delta=\|\ket{\psi''}-\ket{\psi}\ket{\mathcal{F}}\|\le \frac{2\pi n}{2^b}+\epsilon.
	\end{align}
	The total error $\delta$ may be controlled by choosing $b=\Theta(\log{(\frac{n}{\delta})})$ and $\epsilon=\Theta(\delta)$.

	The total $\textsc{T}$ cost of arbitrary state preparation is then then the sum of costs of the data-lookup oracles $f(j,b)=\mathcal{O}(\frac{2^j}{\lambda}+b\lambda)$, adders, and the Fourier state:
	\begin{align}\nonumber
	\textsc{T}\;\text{cost}=&\left(\sum_{j=0}^{n-1}f(j,b)+\mathcal{O}(b)\right)+\mathcal{O}(b\log{(1/\epsilon)})\\\nonumber
	&=\mathcal{O}(2^n/\lambda+nb\lambda+nb+b\log{(1/\epsilon)})\\\nonumber
	&=\mathcal{O}\left(\frac{N}{\lambda}+b(\lambda n + \log{(1/\delta)})\right)\\\nonumber
 &=\mathcal{O}\left(\frac{N}{\lambda}+\log{\left(\frac{\log{N}}{\delta}\right)}\left(\lambda \log{N}+\log{(1/\delta)}\right)\right)
 \\\nonumber
 &=\mathcal{O}\left(\frac{N}{\lambda}+\lambda\log{\left(\frac{\log{N}}{\delta}\right)} \log{\left(\frac{N}{\delta}\right)}\right)
 \\
 &=\mathcal{O}\left(\frac{N}{\lambda}+\lambda \log^2{\left(\frac{N}{\delta}\right)}\right).
	\end{align}

\end{document}